\documentclass[article]{ajs}

\usepackage[utf8]{inputenc}
\usepackage[T1]{fontenc}
\usepackage{amsthm}
\usepackage{amsmath}
\usepackage{amsfonts}
\usepackage{booktabs}
\usepackage{mathtools}
\usepackage{bm}
\usepackage{subcaption}
\usepackage{graphicx}
\usepackage{url}
\usepackage{fixltx2e}
\usepackage[inline]{enumitem}
\usepackage{hyperref}

\newtheorem{proposition}{Proposition}

\DeclareMathOperator{\diag}{diag}

\newcommand{\err}[1]{\varepsilon_t^{#1}\sim\mathcal{N}(0,\sigma_{#1}^2)} 

\author{Krzysztof Rusek\\ AGH University of Science and Technology \\Institute  of Telecommunications \And 
        Mathias Drton\\ Technical University of Munich,\\ Department of Mathematics}

\title{Fine-grained network traffic prediction from coarse data}

\Address{
  Krzysztof Rusek\\
  AGH University of Science and Technology\\
  Institute  of Telecommunications, Al. Mickiewicza 30,\\30-059 Kraków, Poland .\\
  E-mail: \email{krusek@agh.edu.pl}\\
  URL: \url{https://www.agh.edu.pl} \\\\
    Mathias Drton\\
  Technical University of Munich,\\ Department of Mathematics,
  Boltzmannstr.~3, \\85748 Garching b.~München,  Germany\\
  E-mail: \email{mathias.drton@tum.de}\\
  URL: \url{https://www.tum.de}
}

\Plainauthor{Krzysztof Rusek, Mathias Drton} %% comma-separated
\Plaintitle{Fine-grained network traffic prediction from coarse data} %% without formatting
\Shorttitle{Fine-grained forecast from coarse data} %% a short title (if necessary)

\Abstract{
    ICT systems provide detailed information on computer network traffic.  However, due to storage limitations, some of the information on past traffic is often only retained in an aggregated form.  In this paper we show that 
    Linear Gaussian State Space Models yield simple yet effective methods to make predictions based on time series at different aggregation levels.  The models link coarse-grained and fine-grained time series to a single model that is able to provide fine-grained predictions.  Our numerical experiments show up to 3.7 times improvement in expected mean absolute forecast error when forecasts are made using, instead of ignoring, additional coarse-grained observations.  The forecasts are obtained in a Bayesian formulation of the model, which allows for provisioning of a traffic prediction service with highly informative priors obtained from coarse-grained historical data.
}

\Plainkeywords{Network traffic prediction, state space model, Kalman filter, Bayesian structural time series} %% without formatting

\begin{document}

\footnotetext[1]{This work has been submitted to The Austrian Journal of Statistics and is under review process}

%The data and code that support the findings of this study will be openly available in \emph{Algae} GitHub repository at \url{https://github.com/krzysztofrusek/Algae}.

\section{Introduction}

Modern ICT systems produce many different kinds of time series with information about computer network traffic.
%; e.g., in the form of logs or periodic statistics.  
Network administrators often 
%use summaries and logs to 
identify malfunction or Denial of Service Attacks upon visual inspection of the time series.  Numerous open-source and commercial systems have been developed to support them in this task; e.g., \cite{mrtg, munin}.  There has also been much interest in modeling network traffic data in order to make statistical predictions of future traffic. 
% Modeling of such measurements has always been a vivid research topic.
Indeed, the recent development of Software Defined Networking (SDN) allows one to use traffic predictions as one of multiple factors in optimization of network operation \cite{JAGLARZ2020106992}.

A system optimizing network operation must be lightweight software with a small resource footprint.  However, the traffic statistics it uses %yield the basis for predictions
can quickly grow in size as the considered time horizon increases.  Due to storage limitations, it is thus common 
%to aggregate part of the data.
practice to keep only the most recent data in full resolution (for at most a few days). Older observations  are averaged over progressively larger time scales.  Such an aggregation scheme is, for instance, used by the industry standard, high-performance data logging software \emph{RRDtool}~\cite{rrdtool}.

For visual inspection by humans, the progressive decrease in resolution for older observations does not pose a significant problem. 
However, nonuniform resolutions present a challenge for automated quantitive predictions as used in the aforementioned SDNs.
%a researcher or administrator willing to introduce some automation into the logs analysis pipeline, as in the mentioned SDNs, can find such nonuniform representation challenging in modeling. 
In this paper, we propose a general framework to address this challenge using a single Linear Gaussian State Space Model (LG-SSM).
%for a time series with multiple aggregation levels.
%Our application focus is on longer time scales (hours to months) for which 
Specifically, this paper gives
\begin{enumerate*}
	\item a tractable LG-SSM that draws on fine- and coarse-grained data to obtain fine-grained predictions,
	%analytical model of coarse-grained LG-SSM given the model of the underlying fine-grained LG-SSM,
	\item an implementation in a Bayesian framework that  assesses uncertainty in prediction, and
	\item a numerical evaluation of the resulting system based on  real \emph{mrtg} network traffic with a focus on longer time scales (hours to months). %measurements collected in the university network.
\end{enumerate*}  
%under a single Linear Gaussian State Space Model (LG-SSM).
We note that our approach is not a classical fractal-based multi-scale model of the traffic~\cite{abry02multiscale}.

To the authors' best knowledge, the related domain literature on analysis of network traffic does not propose direct competitors to the methods developed in this paper.  In a broader context, %related 
work such as 
% that of 
\cite{Folia:2018aa} addresses aggregation issues but does so with a view towards arbitrary time scales, for which these methods adopt more complicated time discretization and resampling schemes for
%in models based on
continuous-time stochastic processes.  This is in contrast with the applications considered here, where there is a clear finest time scale that is discrete and where aggregation is done in a designed and regular fashion to address storage limitations.
The aggregation schemes we consider also differ from down-sampling schemes in monitoring software as implemented in the MATLAB function \emph{d2d}.  Indeed, down-sampling creates thinned time series whereas we are here considering a system that  retains averages.

The LG-SSM we propose is simple when compared to more involved machine learning methods for traffic prediction~\cite{dcrnn19}.  %However, at least in their current design, these graph neural network methods are limited to a single time scale and do not address our question about benefits of aggregated data. 
However, at least in their current design, these graph neural network methods are limited in three aspects:
\begin{enumerate*}
	\item the methods are limited to a single time scale and do not address our question about benefits of aggregated data,
	%analytical model of coarse-grained LG-SSM given the model of the underlying fine-grained LG-SSM,
	\item they require order of magnitude larger dataset than offered by the \emph{mrtg} tool,
	\item uncertainty in neural networks is hard to estimate and it is a contemporary research topic~\cite{osband2021epistemic}. %measurements collected in the university network.
\end{enumerate*} 
%
%Uncertainty in complicated neural networks is also harder to estimate and it is a contemporary research topic~\cite{osband2021epistemic}.
Moreover, in practice, a simple Fourier decomposition~\cite{rzym20lasso} %(a special case of LG-SSM) 
of the traffic data can help to obtain good forecast accuracy and performance from simple resource-aware models.
A more advanced statistical forecast models are commonly based on ARIMA e.g. \cite{Moayedi2008Arima,Papagiannaki2003longterm} or ARMA~\cite{sang2000}.
Since both the aforementioned models have state-space representation it is possible to extend our proposal to that kind of prediction model.
Having said that, in this work, we focus on structured time series as the decomposition improves the explainability of the prediction. 
The main contribution of this paper is not a new forecasting method but rather a method for improved estimation of the classical models from limited data.

The paper is structured as follows. Section~\ref{sec:methods} outlines the data collection process in the considered network monitoring system, and it develops our joint model for fine and coarse data. Section~\ref{sec:results} presents the setup and the results of our numerical evaluation on traffic in a university network.
Section~\ref{sec:conclusions} summarizes the work, highlights some of its limitations, and suggests simple changes in monitoring systems that would help overcome these limitations.

\section{Methods}\label{sec:methods}
%https://www.springer.com/gp/authors-editors/authorandreviewertutorials/writing-a-journal-manuscript/introduction-methods-and-results/10285524

The Multi Router Traffic Grapher (MRTG~\cite{mrtg}) is an  open-source software for monitoring network link load.
It gathers statistics by pooling SNMP counters from network devices.
This process is continuous and while new measurements arrive the old ones are aggregated. 
The default sampling interval is 5 min. 
There are only 600 such measurements, which cover about 2 days. 
Older results are averaged in groups of six, which corresponds to 30 minute intervals. 
There are again 600 such measurements, which cover about 12 days. 
Measurements older than two weeks (2 days + 12 days) are aggregated in 2h windows. 
There are again 600 of them covering 50 days. The remaining samples are averaged over 24 hours and give a coarse-grained traffic description over 2 years.

\subsection{Linear Gaussian state space model}

In order to model the described time series
%with different segments in  different time resolutions, 
we adopt the framework of Linear Gaussian State Space Models (LG-SSM).  In general, an LG-SSM may be described through a
%We define a Linear Gaussian State Space Model as the
dynamical system~\cite{murphy2012machine}:
\begin{align}\label{eq:state_dynamics}
	\bm z_{t} &= \bm F\bm z_{t-1}+\varepsilon_t, \quad \bm\varepsilon_t\sim \mathcal{N}(\bm b,\bm Q),\\
	\label{eq:state_dynamics2}
	x_t&=\bm H\bm z_t+\bm\delta_t, \quad \bm\delta_t \sim \mathcal{N}(\bm c,\bm R),
\end{align}
where $\bm z_0\sim \mathcal{N}(\bm b_0,\bm Q_0)$ is the $n$-dimensional initial hidden state and $\bm x_t\in\mathbb{R}^m$, $t\ge 0$, is the observed signal. Throughout, $\mathcal{N}(\cdot,\cdot)$ denotes a multivariate Gaussian distribution and the noise vectors $(\bm\varepsilon_t,\bm\delta_t)$ are independent across time $t$.
The state dynamics are parameterized by the  transition matrix $\bm F\in\mathbb{R}^{n\times n}$, the transition noise mean %($n$ 
vector $\bm b\in\mathbb{R}^n $ and covariance matrix $\bm Q\in\mathbb{R}^{n\times n}$.
The observations $x_t$ are noisy linear projections of the states $z_t$ and further parameterized by the observation matrix $\bm H\in\mathbb{R}^{m\times n}$ and the observation noise mean vector $\bm c\in\mathbb{R}^m$ and covariance matrix ${\bm R}\in\mathbb{R}^{m\times m}$.
The system noise $\bm\varepsilon_t$, observation noise $\bm \delta_t$, and the initial state $\bm z_0$ are independent.
%and normally distributed.
Under these assumptions, the Kalman filter can be used for efficient state estimation from noisy observations.
This makes LG-SSM a popular choice for time series forecasting~\cite{harvey1990forecasting}.

For the particular application of interest, we follow~\cite{HU2016forecast} and specify an LG-SSM that incorporates three main components: trend, seasonal effect, and temporally dependent noise.  The trend is modeled by introducing two real hidden states, namely, a current level $l_t$ and a slope $v_t$.  These are real-valued and evolve as
\begin{align}
	\begin{bmatrix}
		l_t\\v_t
	\end{bmatrix}&=\begin{bmatrix}
		1 & 1\\0 & 1
	\end{bmatrix}\begin{bmatrix}
		l_{t-1}\\v_{t-1}
	\end{bmatrix}+\begin{bmatrix}
		\varepsilon_t^{l}\\\varepsilon_t^{v}
	\end{bmatrix}, \label{eq:locallineartrend}
	\end{align}
	where $\err{l}$ and $\err{v}$ are independent.  This model allows for trend changes by a random normal perturbation.  Next, we specify a harmonic signal $f_{t,\omega}$ of frequency $\omega$, which can also be represented as an LG-SSM with known transition matrix with the help of an auxiliary variable $f_{t,\omega}^{*}$.  Specifically,
	\begin{align}
	    	\begin{bmatrix}
		f_{t,\omega}\\f_{t,\omega}^{*}
	\end{bmatrix}&=\begin{bmatrix}
		\cos(\omega) & \sin(\omega)\\-\sin(\omega) & \cos(\omega)
	\end{bmatrix}\begin{bmatrix}
		f_{t-1,\omega}\\f_{t-1,\omega}^{*}
	\end{bmatrix}+\begin{bmatrix}
		\varepsilon_t^{f_\omega}\\\varepsilon_t^{f_{\omega}^*}
	\end{bmatrix}.\label{eq:sezon} 
	\end{align}
	Both components form a hidden state disturbed by independent noise representing phase fluctuations, $\varepsilon_t^{f_{i}^*},\varepsilon_t^{f_{i}}\sim \mathcal{N}(0, \sigma_{f_{\omega}}^2)$. 
Typically the seasonal effect is not harmonic and it is approximated by a sum of multiple Fourier components for both weekly and daily seasonal effects 
%(in our application we use 16 components; 
(see~\cite{rzym20lasso} for details on how to choose the number of components). 
Finally, further temporal dependences are captured via a one-dimensional autoregressive (AR) process
\begin{align}
    %\begin{bmatrix}
		d_t
	%\end{bmatrix}
	&=
	%\begin{bmatrix}
		\alpha 
	%\end{bmatrix}
	%\begin{bmatrix}
		d_{t-1}
	%\end{bmatrix}
	+
	%\begin{bmatrix}
		\varepsilon_t^{d}
	%\end{bmatrix}
	,\label{eq:ar} %\\
	%x_t &= l_t+\sum_i c_{t,i}+d_t+\epsilon_t^{x}\label{eq:struct}
\end{align}
where $\err{d}$.  Summing the three components to one LG-SSM yields the \emph{structural time series} model
\begin{equation}
x_t = l_t+\sum_\omega f_{t,\omega}+d_t+\varepsilon_t^{x}\label{eq:struct},
\end{equation}
where $\err{x}$ is a total error~\cite{harvey1990forecasting}.  In terms of the general representation from \eqref{eq:state_dynamics} we obtain a block-diagonal transition matrix $\bm F$, in which $\alpha$ %from \eqref{eq:ar}
is the only parameter.  The variational inference method applied later ensures that $-1<\alpha<1$ as required for stationarity in \eqref{eq:ar}. 
% {\color{green} This is handled by the surrogate posterior in VI (posterior has form $tanh(\mathcal{N}$) this is the default behaviour in the library}
The observation matrix $\bm H=(1,0,1,0\dots,1)$, with the zeros for states $v_t$ and $f^*_{t,\omega}$, performs the summation in \eqref{eq:struct}. Note that in \eqref{eq:state_dynamics}, $\bm b=0$ and $\bm Q$ is diagonal.  In \eqref{eq:state_dynamics2}, $\bm c,\bm R$ are one-dimensional with $\bm c=0$.

\subsection{Bayesian inference}

In order to fit LG-SSM to the network traffic data of interest, we adopt a Bayesian approach~\cite{murphy2012machine} that captures uncertainty via a posterior distribution.   For our setting this approach has the advantage that today's posterior may play the role of tomorrow's prior distribution.  
The coarse-grained data may thus also serve for the purpose of refining prior distributions %posterior refining
during the provisioning of the prediction service.

For Bayesian inference a joint prior distribution is to be specified for the initial states and the unknown model parameters.  We assume all quantities to be independent a priori, with marginal normal distributions for the initial states, log-normal distributions ($\mathcal{LN}$) for the variance parameters, and a truncated normal distribution for the AR parameter $\alpha$.  The hyperparameters for these distributions are set  using heuristics in TensorFlow Probability (TFP)~\cite{dillon2017tensorflow} that  aim to provide weakly informative prior distributions.
In order to approximate the posterior distribution 
we apply Stochastic Variational Inference~\cite{murphy2012machine}
% {\color{red} better reference?}{\color{green} This should be better}.
Given the parameters, the LG-SSM's predictions (forecasts) are normally distributed.
However, in the Bayesian setup the normality is lost (in general the parameters are not not normally distributed) and the TFP library approximates forecast distributions by mixtures of Normal distributions obtained for a fixed number of posterior samples.

\subsection{Aggregation to coarser time scales}

Let $\bm x_t$ be a time series on a fine time scale. We are interested in leveraging information provided by aggregated versions of $\bm x_t$.  Define the  $r$-aggregated time series $\bm x'_t$ as the average of $r$ consecutive values in non-overlapping windows of length $r$, that is,
\begin{equation}
	\bm x'_t=\frac{1}{r}\sum_{i=0}^{r-1}\bm x_{rt+i}, \quad t=0,1,2,\dots.
	%=\frac{\bm 1}{r}\begin{bmatrix}
	%	x_{r(t-1)+1}&\ldots&x_{rt}\end{bmatrix}.
\end{equation}
Hereafter we use a ${}^\prime$ to denote aggregated time series and their parameters and properties.
Since $r$-aggregation is a linear operation, it retains Gaussianity of $\bm x_t$ for $\bm x'_t$.  In fact, we also retain an LG-SSM as detailed in the following proposition, where we write $\bigoplus_{i=1}^r \bm A_i$ to denote a direct sum producing the block-diagonal matrix $\diag(\bm A_1,\ldots,\bm A_r)$.

\begin{proposition}\label{th:ssm}
	Let $\bm x_t$ follow the LG-SSM 
	%For the stationary state space model 
	from \eqref{eq:state_dynamics}-\eqref{eq:state_dynamics2}.  Then the $r$-aggregated process $\bm x'_t$ follows an LG-SSM
		%linear Gaussian state space model
		%parameterized as follows:
		with $rn$ states given by
		\[
		\bm z'_t=[\bm z_{rt}^T,\bm\varepsilon_{rt+1}^T,\dots,\bm\varepsilon_{rt+r-1}^T]^T,
		\]
		where $[\cdot]^T$ denotes matrix transposition.
        The transition and observation matrices are
\begin{align}
	\bm F' &= 
	%\bigoplus_{i=0}^{r-1} F^{r-i}
 	\begin{bmatrix}
 	\bm F^r & \bm F^{r-1} &\ldots & \bm F\\
 	\bm 0 &\bm 0 & \ldots & \bm 0
 	\end{bmatrix}
	\in\mathbb{R}^{r n\times r n}
	%_{[rn, rn]}
	\label{eq:th:F}, \text{and}\\
	%\bm H'&=\frac{1}{r}\begin{bmatrix}\sum_{i=0}^{r-1}\bm H \bm F^i &\ldots&\bm H +\bm H \bm F & \bm H\end{bmatrix}\\
	\bm H'&=\frac{1}{r}\begin{bmatrix}\sum_{i=0}^{r-1}\bm H \bm F^i,\sum_{i=0}^{r-2}\bm H \bm F^i,\dots,\bm H\end{bmatrix}
% 	\frac{1}{r}\begin{bmatrix}\sum_{i=0}^{k-1}\bm H \bm F^i\mid k=r,\dots,1\end{bmatrix}
	\in \mathbb{R}^{m\times rn}\label{eq:th:H}.
\end{align}
The initial state $\bm z'_{0}$ holds the original initial state and $r-1$ transition noises.
The new system noise $\bm\varepsilon'_{t}=[\varepsilon_{rt}^T,\dots,\varepsilon_{rt+r-1}^T]^T$ is constructed by stacking $r$ system noise vectors, while the new observation noise $\bm\delta'_t$ is an average of the corresponding $r$ observation noise vectors.  They are distributed as
%distributions are given by:
\begin{align}
	%\bm H'=\frac{1}{2}\begin{bmatrix}\bm H +\bm H\bm F& \bm H\end{bmatrix}, \quad
	%\bm F'=\begin{bmatrix}	\bm F^2&\bm F\\\bm 0&\bm 0	\end{bmatrix} \\
	%\bm z'_{k}=\begin{bmatrix}	\bm z_{2k-1}\\\bm \varepsilon_{2k-1}\end{bmatrix}\quad
	\bm\varepsilon'_{t}\sim\mathcal{N}\left([\bm b^T,\ldots,\bm b^T]^T,\bigoplus_{i=1}^{r} \bm Q\right), \quad %\begin{bmatrix} \bm\varepsilon_{2k+1}\\\bm\varepsilon_{2k+2} \end{bmatrix} \quad
	\bm\delta'_t \sim \mathcal{N}\left(\bm c,\frac{1}{r}\bm R\right).%\frac{1}{2}(\delta_{2k-1}+\delta_{2k})	
	\end{align}
\end{proposition}
\begin{proof} 
For clarity, the proof will be given for $r=2$; cases with larger $r$ are analogous.  We have for $t\ge 1$,
\begin{equation}
\begin{bmatrix}\label{eq:matagre}
\bm x'_{t-1}\\
\bm x'_{t}
\end{bmatrix}= \frac{1}{2}\begin{bmatrix}
\bm I & \bm I& \bm 0& \bm 0\\
\bm 0 & \bm 0& \bm I & \bm I
\end{bmatrix}\cdot \begin{bmatrix}\bm x_{2(t-1)}\\\bm x_{2(t-1)+1}\\\bm x_{2t}\\\bm x_{2t+1}\end{bmatrix},
\end{equation}
where $\bm I,\bm 0 \in \mathbb{R}^{m\times m}$ are the identity matrix and the matrix with all entries  $0$, respectively.
% The sequence of four observations is totally determined by the first hidden state and all the noises. By 
Unrolling the dynamics from~\eqref{eq:state_dynamics} we obtain that
\begin{equation*}
\begin{bmatrix}\bm x_{2t-2}\\\bm x_{2t-1}\\\bm x_{2t}\\\bm x_{2t+1}\end{bmatrix}=
\begin{bmatrix} \bm H & \bm 0 &\bm 0 & \bm 0 \\
				\bm H\bm F & \bm H & \bm 0 & \bm 0 \\
				\bm H\bm F^2&\bm H\bm F & \bm H & \bm 0 \\
				\bm H\bm F^3&\bm H\bm F^2&\bm H\bm F & \bm H & \\
\end{bmatrix} \begin{bmatrix}\bm z_{2t-2}\\\bm\varepsilon_{2t-1}\\\bm\varepsilon_{2t}\\\bm\varepsilon_{2t+1} \end{bmatrix}+
\begin{bmatrix}\bm\delta_{2t-2}\\\bm\delta_{2t-1}\\\bm\delta_{2t}\\\bm\delta_{2t+1} \end{bmatrix}.
\end{equation*}
Hence,
%Using matrix product associative property in~\eqref{eq:matagre} an the rules of block matrix algebra we get:
\begin{align}
\begin{bmatrix}
\bm x'_{t-1}\\
\bm x'_{t}
\end{bmatrix}=&\frac{1}{2}\begin{bmatrix}
\bm H +\bm H\bm F & \bm H& \bm 0 & \bm 0\\\nonumber
\bm H\bm F^2+ \bm H\bm F^3& \bm H\bm F+\bm H\bm F^2 & \bm H +\bm H\bm F & \bm H
\end{bmatrix} \times
%\end{bmatrix} \\ &\times
\begin{bmatrix}\bm z_{2t-2}\\\bm\varepsilon_{2t-1}\\\bm\varepsilon_{2t}\\\bm\varepsilon_{2t+1} \end{bmatrix}\\&+
\frac{1}{2}\begin{bmatrix}\bm\delta_{2t-2}+\bm\delta_{2t-1}\\\bm\delta_{2t}+\bm\delta_{2t+1} \end{bmatrix}, %\begin{bmatrix}x_{2k-1}&x_{2k}&x_{2k+1}&x_{2k+2}\end{bmatrix}^T
\end{align}
which simplifies to
\begin{equation}\label{eq:unroledssprime}
\begin{bmatrix}
\bm x'_{t-1}\\
\bm x'_{t}
\end{bmatrix}=\begin{bmatrix} \bm H' & \bm 0 \\ \bm H'\bm F' & \bm H'\end{bmatrix}
\begin{bmatrix} \bm z'_{t-1} \\ \bm\varepsilon'_{t}\end{bmatrix} +
\begin{bmatrix} \bm\delta'_{t-1} \\ \bm\delta'_{t} \end{bmatrix}.
\end{equation}
Equation~\eqref{eq:unroledssprime} defines a new state space model as
\belowdisplayskip=-12pt
\begin{align*}
	\bm z'_{t}&=\bm F'\bm z'_{t-1}+\bm \varepsilon'_t,\\
		\bm x'_t& = \bm H'\bm z'_t+\bm\delta'_t.
\end{align*}\qedhere
% where $\bm z'_{k}=\begin{bmatrix}	\bm z_{2k-1}^T&\bm \varepsilon_{2k-1}^T\end{bmatrix}^T$ is the new state and that concludes the proof.
\end{proof}

\section{Numerical experiment}
%Results}
\label{sec:results}

Proposition~\ref{th:ssm} links observations at a different aggregation level ($\bm x'_t$) to the single underlying LG-SSM.
This link %($\agg$ function) 
allows the construction of a joint model for two or more time scales that, as we demonstrate now, yields improved network traffic predictions by leveraging aggregated historical information.
% strength of the estimated model.
In our Bayesian approach to statistical inference we work with simple vague priors that assume prior independence of parameters and are set using heuristics TensorFlow Probability (TFP)~\cite{dillon2017tensorflow}, as described in Section~\ref{sec:methods}.  For our experiments the following priors were obtained. At the finest time scale, the initial states $l_0$ and $d_0$ are normal with mean $-219$ and variance $5\cdot 10^5$.  The states $v_0$ and $f_0$ are centered normal with variances $10^5$ and $3\cdot 10^5$, respectively.  The log-normal variances $\sigma_l,\sigma_v,\sigma_d,\sigma_{f_\omega}, \sigma_x$ are all derived from a normal distribution with variance 3.  The means are 2.8 for the first three variances, and 1.2 for the rest.
The AR parameter $\alpha$ is truncated standard normal, with truncation to $(-1,1)$.  These same priors are also used to induce priors for aggregated coarse-grained time series; recall Proposition~\ref{th:ssm}, where the state $z'_t$ combines a fine-scale initial state with system noise, the joint distribution of which is determined by the fine-scale priors and dynamics.

Our numerical experiment uses a real MRTG dataset from %university
downlink traffic measurements at the first author's university.
In the experiment, we use 600 samples taken every 30 min (fine-grained) and 600 samples taken every 2h (coarse-grained), so in this case $r=4$.
%{\color{blue} 
While the model can use all four aggregation levels in the data, we focus here on only two of them  because the practical solutions proposed in the literature are focused on a time scale of about 1h~\cite{JAGLARZ2020106992} and because we could substantially simplify our implementation for $r$ being a power of two.
The overall time span then covers around 62 days, from which the last 2 days are used only for forecast validation.
The data exhibits a strong seasonal component  including both weekly and daily patterns (see Figure~\ref{fig:ts}), which our model picks by inclusion of two periodic components (16 harmonics each).
\begin{figure*}[t] %  figure placement: here, top, bottom, or page
	\centering
	\includegraphics[width=\textwidth]{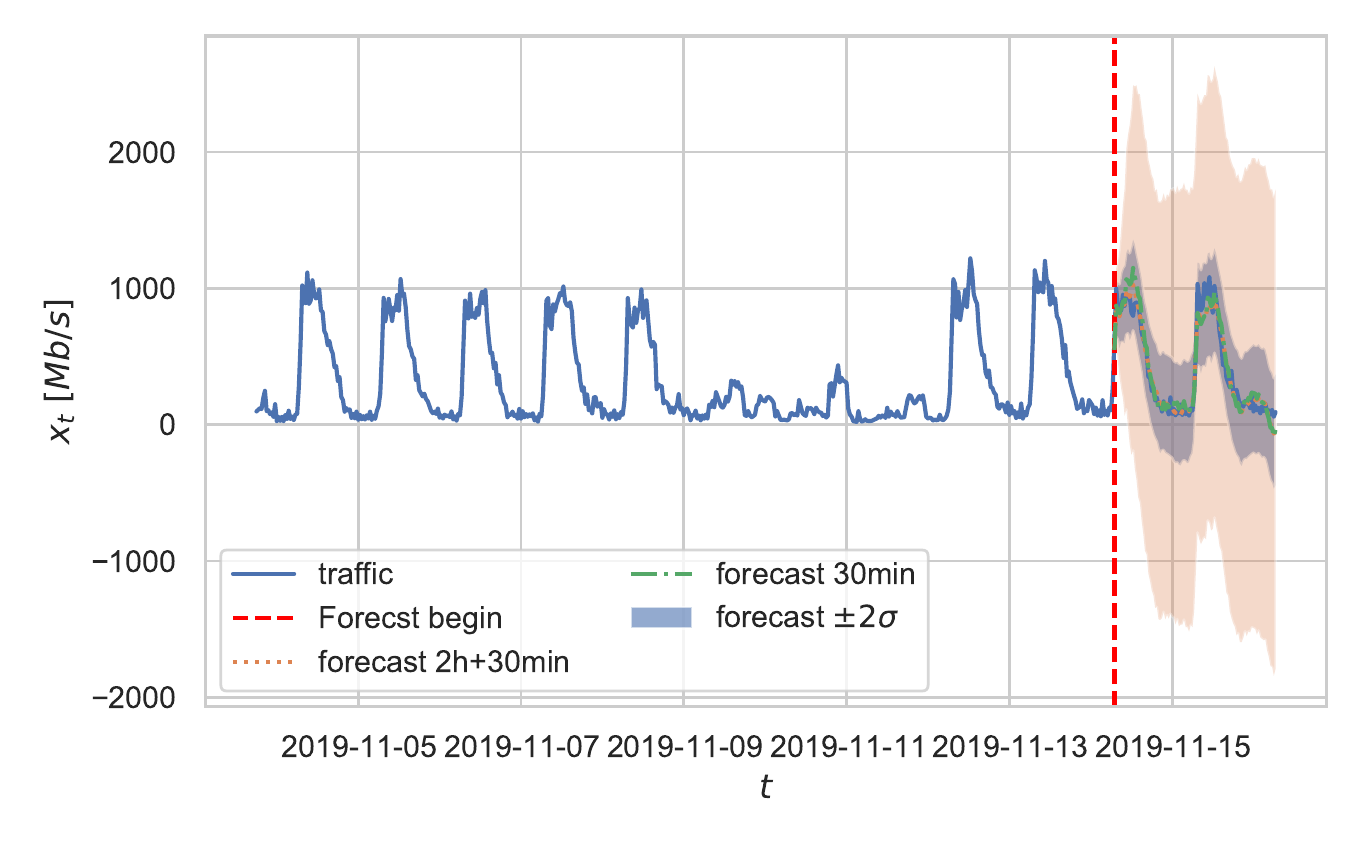}
	\caption{Fine-grained traffic time series and prediction.  The orange ribbon represents uncertainty for the fine-grained model only. The blue ribbon is obtained from our joint model.  Additional observations improve the likelihood of forecast by a factor of $e^{2.32}\approx 10$ (log-likelihood changed from 543.6 to 545.92) and expected $\textrm{mae}$ over \textbf{3.7$\times$}. Holiday: 11, weekend: 9 10 November.  }
	\label{fig:ts}
\end{figure*}
We note that since the fine-grained part of the series contains only one weekend and a holiday on Monday, it is
%this makes it 
difficult to estimate the weekly seasonal effect as well as long-range trend, as represented by~\eqref{eq:locallineartrend}, using only observations from one time scale.
The forecast from the fine-grained model based on~\eqref{eq:struct} is presented in Figure~\ref{fig:ts} and denoted as \emph{forecast 30min}.
The forecast distribution is a mixture of 50 equally probable Kalman prediction distributions obtained for 50 samples from the surrogate posterior in the variational approach.

Although the average is reasonably accurate, the forecast uncertainty is very large.
This can be explained by the fact that the extra %day-of of 
Monday breaks the periodic patterns and increases variability in the trend component.
Addition of additional observations at the coarser 2h resolution substantially improves the prediction.
The prediction interval decreases by more than a factor or three, and the likelihood of the forecast data increases about ten times.

Further intuition about the results can be provided by computing the expected mean absolute error, where the expectation is computed for estimated forecast distribution. 
In other words, this is the average error for all possible future values and thus it depends on the entire forecast distribution.
%(602.3754097639135, 36.24531543199165, 'TimeSeriesPredictor')
%(160.50127, 8.986970635505978, 'AggregatedPredictor')
The expectation computed with the Monte Carlo method using 100 samples (possible forecasts) yield the following values: $\mathbf{\mathrm{mae}_{30min}=602\pm 36}$ Mb/s and $\mathbf{\mathrm{mae}_{30min+ 2h}=160\pm 9}$ Mb/s.
The use of historical data decreased $\mathrm{mae}$ over \textbf{3.7$\times$}.
We emphasize that this substantial improvement in model accuracy comes at a moderate computational cost. 

As previously mentioned, the estimation and forecast were implemented using TensorFlow Probability~\cite{dillon2017tensorflow}.
The implementation with additional numerical examples will be published as open-source software.
The \emph{LinearOperator} API provided by TensorFlow allows one to implement the model efficiently without instantiating the large block matrices from~\eqref{eq:th:F} and~\eqref{eq:th:H}.
Only small dense matrix multiplications are executed, which can be done efficiently via the associative scan procedure. 
Furthermore, our code was
JIT-compiled using XLA -- a TensorFlow linear algebra compiler and executed on GPU.
We observed increased numerical precision of compiled (and optimized) computations compared to step-by-step execution.
The computational cost of 300 variational steps of the estimation is about 2.5h on Tesla V100 SXM2 GPU.
Most of the time is used for compilation (we observed a similar run time without JIT).
However,  as described earlier this cost is amortized by the fact that once we have determined the parameter posterior we can refine the model online without lengthy computations.

\section{Discussion and Conclusions}\label{sec:conclusions}

In this paper we proposed a Bayesian structural time series model for network traffic measurements that feature different time scales. Our LG-SSM approach readily allows one to aggregate fine- versus coarse-grained observations. Our experiment shows that including older coarse-grained statistics on network traffic can drastically improve prediction accuracy and uncertainty.  We are confident this observation would similarly hold in other modeling context, including for example single source of error models (aka exponential smoothing)~\cite{forecast}.

The data we consider were collected using the MRTG tool and come in the form of nonnegative measurements.  The fact that we model these directly as Gaussian and with additive structure is a clear limitation of our work.  Indeed, log-normal distributions have been shown to give more accurate descriptions of network traffic~\cite{lognormtraffic}.  When working with a single time scale this issue can be addressed quite straightforwardly by applying our model to the log-transformed traffic data.
With this transformation additive effects in the model would correspond to multiplicative effects on the original nonnegative scale, and thus more appropriately capture the fact that during the periods of highest traffic intensity (middle of the day) one observes higher fluctuations compared to nights when the traffic is smallest.  
Classical alternative to this transformation is GARCH model~\cite{garch}.
However, with more than one time scale and aggregation of data the use of log-transformations or GARCH is more subtle because time aggregation by averaging is no longer on the log-scale.  As a resolution of this problem we propose that software systems store geometric instead of arithmic means when aggregating older measurements of traffic size.  This simple modification would make our model (and other possible LG-SSM) directly applicable to log-transformed data with multiple time scales.

We tested this proposal on the real traffic, see detail in the Appendix~\ref{sec:geo}.
The experimental results support our claims for both log-transformed and non-scaled data.
The reduction of prediction variance significantly improves predictions in the log-transformed domain. 
The application of log transform also reduces forecast error by eliminating unrealistic negative forecasts. We observed a 4\% reduction of mae, solely due to forecast in the log domain. This is a clear indication that when it comes to network traffic, the geometric average is more informative compared to arithmetic aggregation.

\section{Acknowledgments}
This work was supported by the National Science Centre, Poland under grant nr 2019/03/X/ST7/00386, Polish Ministry of Science and Higher Education with the subvention funds of the Faculty of Computer Science, Electronics and Telecommunications of AGH University and by the PL-Grid Infrastructure.

\appendix

\section{Geometric averaging}\label{sec:geo}
In this section, we report the results of an experimental aggregator based on geometric averaging.
Since MRTG offers only 600 raw samples at a resolution of 5 min it cannot be used to evaluate geometric aggregation over 2h and 30 min as this requires 18,000 samples.

For the purpose of evaluating the link was monitored for over two months to obtain a new time series.
The traffic was aggregated into 600 samples at resolution 30 min and 600 at resolution 2h to match exactly our main experiment.
Two aggregated datasets were produced: one from the raw observation (arithmetic), and one from the log-transformed observations (geometric).
For both datasets, we repeated the main experiment reported in the paper obtaining two types of forecast: trained on 30 min data only and combined 2h and 30 min.
The Mean absolute errors obtained in the four experiments are reported in table~\ref{tab:mae}.
\begin{table}[h]
  \caption{Mean Absolute Error (Mb/s) of the forecast.}
  \label{tab:mae}
  \centering
  \begin{tabular}{ccl}
    \toprule
    Aggregation &30 min&30 min and 2 h\\
    \midrule
    Arithmetic & 395$\pm$1.8 & 389$\pm$1.8\\
    Geometric & 2749$\pm$407& \bf{376}$\pm$5\\
  \bottomrule
\end{tabular}
\end{table}
The accuracy is calculated in the original scale.
The forecasts in the geometric case are transformed forecasts from the log-transformed observations and visualized in figure~\ref{fig:tsgeo}.
\begin{figure}[h!] %  figure placement: here, top, bottom, or page
	\centering
	\includegraphics[width=\textwidth]{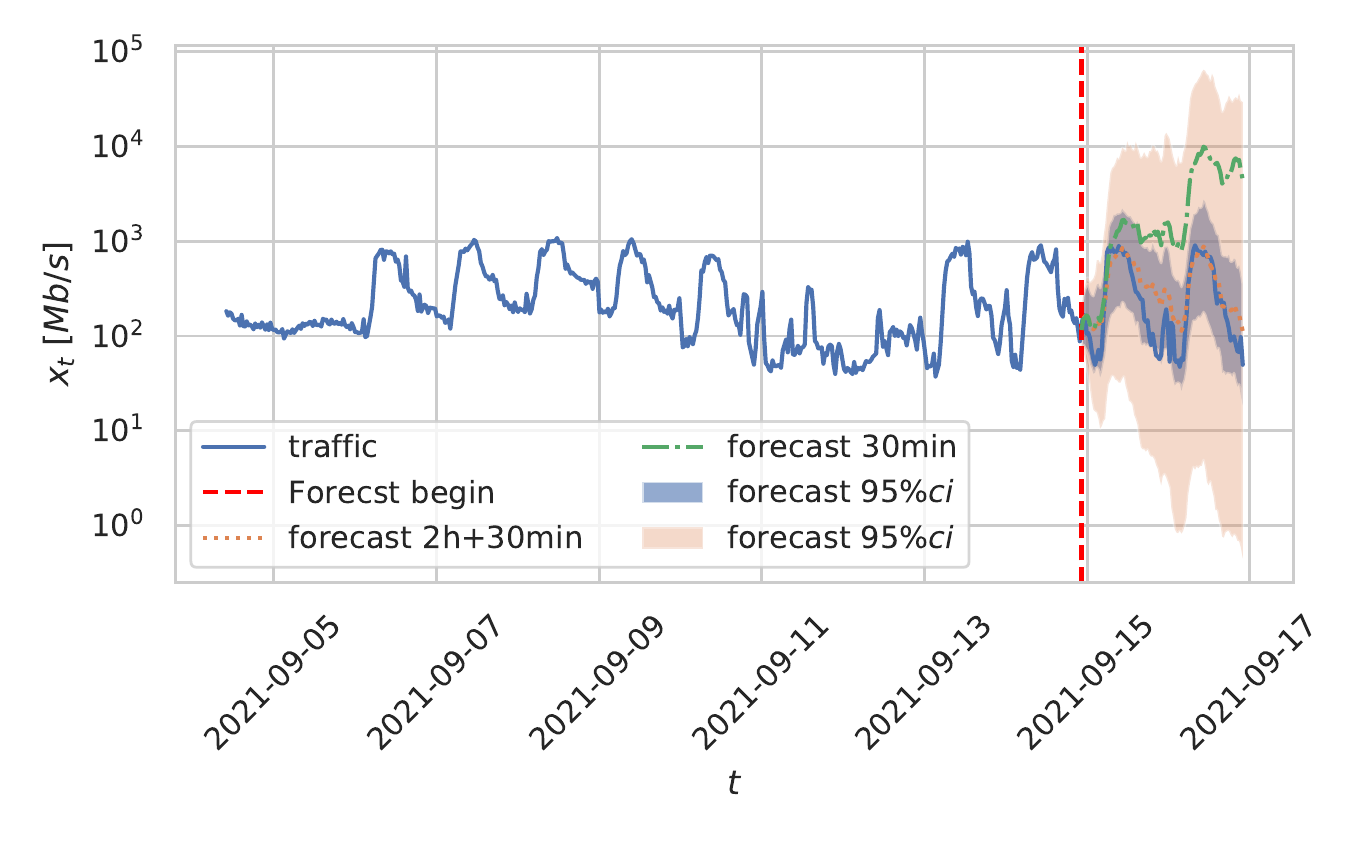}
	\caption{Fine-grained traffic time series and prediction for geometric aggregation. }
	\label{fig:tsgeo}
\end{figure}

The first observation from the experiment is that additional observations always improve prediction.
In the arithmetic experiment, the benefit is not large as in the main experiment because the traffic is quite regular and the training dataset does not contain any holidays that may introduce a large variance of the trend component.
Having said that in the geometric experiment we observe a huge improvement in the prediction.
In fact, the additional observations at 2h resolution are necessary to make a long-term forecast.
Without this, the large variance of the model gets magnified by the $\exp$ transform that results in the average (note that the mean of log-normal distribution depends on the variance of the underlying normal distribution) of over an order of magnitude too large in two days.
Using a median as a point forecast would yield more stable and interpretable long-term predictions as with normal distribution, the mean equals the median. And the median of a transformed random variable is the transformed median.
On the other hand, low variance forecasts from the model trained on the joint dataset follow the observed traffic.
Furthermore, the forecast error is reduced by ~5\% compared to the baseline of a model trained on 30 min dataset with arithmetic aggregation.

\end{document}